\documentclass[aps,prl,twocolumn,superscriptaddress,nofootinbib,floatfix]{revtex4-2}
\usepackage{amsmath,amssymb,amsthm,bm,graphicx,hyperref,xcolor,soul}
\hypersetup{colorlinks=true,linkcolor=blue!60!black,citecolor=blue!60!black,urlcolor=blue!60!black}

\usepackage{bbm}

\def \beq {\begin{equation}}
\def \eeq {\end{equation}}

\newcommand{\Tr}{\operatorname{Tr}}
\newcommand{\supp}{\operatorname{supp}}
\newcommand{\E}{\mathbb{E}}
\newcommand{\Pn}{\mathcal{P}_n}
\newcommand{\I}{\mathcal{I}}
\newcommand{\Ot}{\mathcal{O}}
\newcommand{\F}{\mathbb{F}}
\newcommand{\1}{\mathbbm{1}}
\newcommand{\Choi}[1]{|#1\rangle\!\rangle}
\def \Oc {\mathcal{O}}
\newtheorem{theorem}{Theorem}
\newtheorem{proposition}{Proposition}
\newtheorem{lemma}{Lemma}

\newtheorem{conjecture}{Conjecture}
\theoremstyle{remark}

\begin{document}

\title{Pauli instability in arbitrary states: detecting magic in physical correlators}

\author{Tanner Jackson}
\affiliation{Department of Physics and Astronomy, University of Kentucky, Lexington, KY 40506, USA}
\author{Alexey Milekhin}
\affiliation{Department of Physics and Astronomy, University of Kentucky, Lexington, KY 40506, USA}
\footnotetext{t.jackson@uky.edu, milekhin@uky.edu}

\begin{abstract}
Measures of non-stabilizerness (``magic'') of a quantum evolution quantify how hard it is to simulate on a classical computer. A natural question is how to extract information about magic from physical correlators, for example at finite temperature or in the ground state.
In this paper we show how to measure magic using out-of-time-order correlators evaluated in an arbitrary state. For every full-rank state, and in particular for every Gibbs state, the resulting quantity vanishes if and only if the evolution is Clifford. We establish its basic properties, relate it to the stabilizer entropy of the Choi state, and prove that, even with a single fixed probe operator, it gives a lower bound on the number of $T$ gates needed to implement the dynamics in any circuit architecture.
\end{abstract}

\maketitle

\emph{Introduction.}---Clifford circuits acting on stabilizer states can be simulated efficiently on a classical computer~\cite{Gottesman1998,AaronsonGottesman2004}; the resource that must be added to reach universality is non-stabilizerness, or \emph{magic}~\cite{BravyiKitaev2005,Veitch2014,HowardCampbell2017}. 
Among the many magic measures, the stabilizer R\'enyi entropies (SRE)~\cite{Leone2022} have recently received a lot of attention for many-body applications because they are expressed through Pauli expectation values and are therefore computable and, in principle, measurable~\cite{Oliviero2022,HaugPiroli2023}. The SRE can be extended from states to dynamics \cite{Dowling2024} via the Choi state leading to a sum over infinite-temperature two-point functions $\Tr[A(t)B]/2^n$, $A(t)\equiv U^\dagger AU$. Why this detects magic is transparent: if $U$ is Clifford, $A(t)$ is again a Pauli string and the normalized correlator is $1$ or $0$ for every pair $(A,B)$.

This construction has an obvious drawback if one wants to detect magic in a laboratory or in a numerical simulation: preparing an infinite temperature state in the Choi form requires a maximally entangled state of $2n$ qubits; moreover numerical evolution of pure states is exponentially faster compared to mixed states. The natural alternative, $\Tr[\rho A(t)B]$ with a thermal $\rho$, fails at the first step: for Clifford $U$ this is the expectation value of a Pauli string in $\rho$, which is generically not $\pm1$ or $0$, so no Clifford/non-Clifford dichotomy is left.

Recently, refs.~\cite{Garcia2023,Garcia2024} proposed to diagnose magic through the out-of-time-ordered correlator (OTOC) through the ``Pauli instability'':
\beq
\I = -\log\E_{A,B \in \Pn} \frac{1}{2^n}|\Tr(A(t) B A(t) B)|, 
\eeq
where $\Pn$ is the n-qubit Pauli group and $\E$ is the uniform average.

\textit{The key observation of the present work is that the mechanism now survives in arbitrary states}: for Clifford $U$ the operator $A(t)BA(t)B$ is $\pm\1$, so the OTOC
\begin{equation}
    \Ot_\rho(A,B)\equiv \Tr\!\big[\rho\,A(t)\,B\,A(t)\,B\big],
    \label{eq:otoc}
\end{equation}
is equal to $\pm 1$, \emph{for every state~$\rho$}. We therefore define, for $\alpha>0$,
\begin{equation}
    \I_\alpha(\rho,U)\equiv-\log\,\E_{A,B\in\Pn}\big|\Ot_\rho(A,B)\big|^\alpha ,
    \label{eq:def}
\end{equation}
Notice that this OTOC has a single insertion of $\rho$, so it is different from the one appearing on the bound on chaos \cite{MSS2016} or thermal operator size \cite{Qi:2018bje,Lucas:2018wsc}. The last fact is especially interesting because a lot of our proofs will be based on the effective operator "thermal size".

In this Letter we study $\I_\alpha(\rho,U)$ analytically. Our main results are the following. (i)~For every state $\rho$, $\I_\alpha(\rho,U)$ is non-negative, bounded by $\simeq2n\log2$, covariant under Clifford unitaries, additive over subsystems, and monotone in $\alpha$ (Proposition~\ref{prop:basic}). (ii)~For every \emph{full-rank} $\rho$, in particular for any Gibbs state, $\I_\alpha(\rho,U)=0$ iff $U$ is Clifford (Theorem~\ref{thm:faithful}); the full-rank condition cannot be dropped. (iii)~At $\alpha=2$ the average over $B$ can be done in closed form, Eq.~\eqref{eq:alpha2}: $e^{-\I_2}$ is a sum of thermal two-point functions $|\langle QA(t)\rangle_\rho|^2$ weighted by the Pauli spectrum of $A(t)$, and at infinite temperature it reduces to the stabilizer entropy of the Choi state $\Choi U$. (iv)~For every circuit made of Clifford gates and $k$ $T$~gates, in \emph{any} arrangement, and every $\rho$,
\begin{equation}
    \I_2(\rho,U)\le k\log 2 ,
    \label{eq:tcount}
\end{equation}
so that $\I_2(\rho,U)/\log2$ is a lower bound on the $T$-count of the dynamics (Theorem~\ref{thm:tcount}). (v)~The same bound holds, and becomes tight, when the average over $A$ is replaced by a single fixed probe operator (Theorem~\ref{thm:fixedA}), which reduces the number of correlators from $16^n$ to $4^n$ and matches the experimental situation of a fixed observable. 
Also, we will consider several examples.
(vi)~We compute $\I_\alpha$ exactly for $T$ gates acting in parallel, Eq.~\eqref{eq:parallelT}, which exhibits how the $Z$-correlations of the state hide magic and motivates the conjecture that the sharp constant in \eqref{eq:tcount} is $\log\frac43$; we bound $\I_\alpha$ by $n\log2$ for the third level of the Clifford hierarchy; we evaluate it for Wigner-matrix dynamics, where it grows as $4t^2$ at early times; and we estimate it numerically in various states for integrable and non-integrable forms of the Heisenberg XXX model. Proofs are collected in the Supplemental Material (SM).

\emph{Basic properties.}---Throughout, $\rho$ is an \emph{arbitrary} density matrix on $n$ qubits, $d=2^n$, and $U$ is any unitary. The following properties hold without any assumption on $\rho$ (SM, Sec. \ref{app:basic}):
\begin{proposition}\label{prop:basic}
For every density matrix $\rho$, every unitary $U$ and every $\alpha>0$:
\begin{enumerate}
\item[(i)] \emph{Range.}
\begin{equation}
    0\le\I_\alpha(\rho,U)\le 2n\log2-\log\big(2-4^{-n}\big) .
    \label{eq:range}
\end{equation}
\item[(ii)] \emph{Dependence on $\alpha$.} For $0<\alpha\le\alpha'$,
\begin{equation}
    \I_\alpha\le\I_{\alpha'}\le\tfrac{\alpha'}{\alpha}\,\I_\alpha .
    \label{eq:alphamono}
\end{equation}
\item[(iii)] \emph{Clifford covariance.} For Clifford unitaries $V_1,V_2$,
\begin{equation}
    \I_\alpha(\rho,\,V_1UV_2)=\I_\alpha(V_2\rho V_2^\dagger,\,U).
    \label{eq:covariance}
\end{equation}
\item[(iv)] \emph{Additivity.} $\I_\alpha(\rho_1\otimes\rho_2,U_1\otimes U_2)=\I_\alpha(\rho_1,U_1)+\I_\alpha(\rho_2,U_2)$.
\item[(v)] \emph{Clifford evolutions are free.} If $U$ is Clifford, $\I_\alpha(\rho,U)=0$.
\end{enumerate}
\end{proposition}
 Property (iii) says that a Clifford applied after the evolution is free, while one applied before is absorbed into the state; (iv) and (v) are immediate from the definition, (v) being the mechanism explained in the introduction.
The converse of (v) is the central property, and it is the only one that requires an assumption on the state (SM, Sec. \ref{app:faith}):
\begin{theorem}[Faithfulness]\label{thm:faithful}
Let $\rho$ have full rank. Then $\I_\alpha(\rho,U)=0$ if and only if $U$ is a Clifford unitary (up to a phase).
\end{theorem}
 The full-rank assumption cannot be dropped: for $\rho=|0\rangle\langle0|$ and $U=T$ one finds $\I_\alpha=0$, because $T$ acts trivially on its own eigenstate. We note that $\I_\alpha$ is not monotone under stabilizer operations acting on the state alone: for the single-qubit example above, $\I_\alpha(|0\rangle\langle0|,T)=0$ while $\I_\alpha(|{+}\rangle\langle{+}|,T)=\log\frac43$, and $|{+}\rangle=H|0\rangle$. This is expected, since the state is not the resource here.

\emph{Structure at $\alpha=2$.}---
We will mostly work with $\alpha=2$, where the Pauli average over $B$ can be performed in closed form.
Writing $A(t)=\sum_{Q\in\Pn}c_Q(A)\,Q$ we quote the final result
(SM, Sec. \ref{app:Baver}):
\begin{equation}
    \E_B\big|\Ot_\rho(A,B)\big|^2=\sum_{Q}c_Q(A)^2\,\big|\langle Q\,A(t)\rangle_\rho\big|^2 ,
    \label{eq:alpha2}
\end{equation}
where $\langle\,\cdot\,\rangle_\rho=\Tr[\rho\,\cdot\,]$. Thus $e^{-\I_2(\rho,U)}$ is an average over $A$ of thermal two-point functions $\langle QA(t)\rangle_\rho$ of Pauli strings, weighted by the infinite-temperature ``Pauli spectrum'' $c_Q(A)^2$ of $A(t)$. At infinite temperature $\langle QA(t)\rangle_{\1/d}=c_Q(A)$ and
\begin{equation}
    \I_2(\1/d,U)=-\log\E_A\sum_Qc_Q(A)^4=M_2\big(\Choi{U}\big),
    \label{eq:choi}
\end{equation}
the stabilizer 2-R\'enyi entropy of the Choi state~\cite{Leone2022,Dowling2024}. Equation~\eqref{eq:alpha2} is therefore a finite-temperature deformation of the operator SRE, and Eq.~\eqref{eq:alphamono} shows that the same is true of $\I_1$ up to a factor of two.

\emph{Lower bound on the $T$-count.}---Any circuit of Clifford gates and $k$ $T$ gates, with the $T$ gates on arbitrary qubits and arbitrarily interleaved with Cliffords, can be brought to the normal form $U=C\,e^{-i\pi P_k/8}\cdots e^{-i\pi P_1/8}$ with $C$ Clifford and $P_j$ signed Pauli strings, by commuting the Cliffords to the left. Conjugation by $e^{-i\pi P/8}$ leaves a Pauli string commuting with $P$ unchanged and maps an anticommuting one to $(Q-iQP)/\sqrt2$. Hence the support of $A(t)$ lies in the coset $A\langle P_1,\dots,P_k\rangle$ and
\begin{equation}
    N_A\equiv\big|\supp A(t)\big|\le 2^{r}\le 2^k ,
\end{equation}
where $r$ is the rank of $\{P_j\}$ viewed as vectors in $\F_2^{2n}$. Averaging over $A$ we obtain (SM, Sec. \ref{app:op_size}):
\begin{theorem}\label{thm:tcount}
For every density matrix $\rho$ and every Clifford${}+T$ circuit $U$ with $k$ $T$ gates,
\begin{equation}
    \I_2(\rho,U)\;\le\;-\log\E_A\frac{1}{N_A}\;\le\;r\log2\;\le\;k\log2 .
    \label{eq:tcount2}
\end{equation}
If the $k$ Pauli axes $P_j$ are linearly independent over $\F_2$ (in particular $k\le2n$), the bound improves to $\I_2(\rho,U)\le k\log\frac32$.
\end{theorem}
For small values of $k$, one can exhaustively check all possible dependency relations among the Pauli axes $P_j$ and find that $\I_2 (\rho,U) \leq k \log \frac{3}{2}$ holds generally for $k \leq 7$. In addition, nothing in the argument uses the angle $\pi/8$: Eq.~\eqref{eq:tcount2} holds for circuits with $k$ Pauli rotations $e^{i\theta_jP_j}$ of arbitrary angles. 

The $k\log2$ in \eqref{eq:tcount} is not the best constant. For $k$ $T$ gates acting in parallel on distinct qubits (SM, Sec. \ref{app:Tpar}),
\begin{equation}
    \I_\alpha\big(\rho,T^{\otimes k}\otimes\1\big)=k\log\tfrac43-\log\sum_{S\subseteq[k]}3^{-|S|}\,\big|\langle Z_S\rangle_\rho\big|^\alpha ,
    \label{eq:parallelT}
\end{equation}
where $Z_S=\prod_{i\in S}Z_i$ and the $S=\emptyset$ term equals~$1$. Hence $\I_\alpha(\rho,T^{\otimes k})\le k\log\frac43$, with equality iff all $\langle Z_S\rangle_\rho$, $S\neq\emptyset$, vanish (e.g.\ at infinite temperature), and $\I_\alpha=0$ iff the $k$ qubits are in a computational-basis product state. Equation~\eqref{eq:parallelT} exhibits explicitly how a finite-temperature state reduces the visible magic through its $Z$-correlations.

Refs. \cite{Garcia2024} only considered special arrangements of $T$ gates above. Here we prove a bound for any architecture. 
\begin{theorem}[Tight bound for maximally mixed state]\label{thm:tcount_tight}
For every Clifford${}+T$ circuit $U$ with $k$ $T$ gates,
\begin{equation}
    \I_2(\1/d,U) \le k \log \tfrac{4}{3}
    \label{eq:tcount43}
\end{equation}
\end{theorem}
This theorem is proven in (SM, Sec. \ref{app:proof_of_infT}). By virtue of eq. (\ref{eq:choi}) this provides a tight bound on the stabilizer R\'enyi entropy.
It would be very interesting to generalize this bound for arbitrary $\rho$, however we have not managed to do that. Hence we formulate it as a conjecture:
\begin{conjecture}\label{conj}
For every $\rho$ and every Clifford${}+T$ circuit with $k$ $T$ gates, $\I_2(\rho,U)\le k\log\frac43$.
\end{conjecture}
We have numerically analyzed this inequality for thousands of circuits with adversarially optimized states without finding any violations. 

\emph{Fixed probe operator.}---The double Pauli average in \eqref{eq:def} involves $16^n$ correlators, the square of the cost of a single Pauli average, and in an experiment one typically has access to a fixed probe operator $A$ rather than to all of them. It is therefore natural to define
\begin{equation}
\begin{split}
    \I_2(\rho,U;A)&\equiv-\log\E_B\big|\Ot_\rho(A,B)\big|^2\\
    &=-\log\sum_Qc_Q(A)^2\big|\langle QA(t)\rangle_\rho\big|^2 ,
\end{split}
    \label{eq:fixedA}
\end{equation}
so that $e^{-\I_2(\rho,U)}=\E_A\,e^{-\I_2(\rho,U;A)}$. The results above descend to a single probe, and the $T$-count bound becomes sharp (SM, Sec. \ref{app:op_size}):
\begin{theorem}[Fixed probe]\label{thm:fixedA}
Let $\rho$ be any density matrix and $A\in\Pn$ any probe.
\begin{enumerate}
\item[(i)] $0\le\I_2(\rho,U;A)\le\log N_A$, where $N_A=|\supp U^\dagger AU|$. If $\rho$ has full rank, $\I_2(\rho,U;A)=0$ iff $U^\dagger AU$ is a Pauli string.
\item[(ii)] If $U$ is a Clifford${}+T$ circuit with $k$ $T$ gates whose axes have $\F_2$-rank $r$, then
\begin{equation}
    \I_2(\rho,U;A)\le r\log2\le k\log2 .
    \label{eq:fixedAbound}
\end{equation}
\item[(iii)] The bound is tight: for $U=T^{\otimes k}$ and $A=X^{\otimes k}$,
\begin{equation}
    e^{-\I_2(\rho,U;A)}=2^{-k}\sum_{S\subseteq[k]}\big|\langle Z_S\rangle_\rho\big|^2 ,
    \label{eq:fixedAtight}
\end{equation}
so that $\I_2(\rho,U;A)=k\log2$ at infinite temperature and, more generally, whenever $\langle Z_S\rangle_\rho=0$ for all non-empty $S\subseteq[k]$.
\end{enumerate}
\end{theorem}
 In contrast to the averaged quantity, the constant $\log2$ per $T$ gate is the best possible for a fixed probe, whereas Conjecture~\ref{conj} asserts $\log\frac43$ for the average: the factor $\frac34$ reflects the fact that a $T$ gate with axis $P$ leaves the half of all Pauli strings commuting with $P$ untouched, a property of the average over $A$ and not of individual probes. The single average over $B$ can of course also be sampled, as discussed above.

\emph{Clifford hierarchy.}---
In this section we ask how much Pauli instability can change under the application of a \textit{single} unitary in the third level of the Clifford hierarchy~\cite{GottesmanChuang1999}. 
Unitaries in the third level, $U^\dagger\Pn U\subseteq\mathcal C_n$ (e.g.\ $T$, CCZ, Toffoli), are special because $V_A\equiv A(t)$ is then a Hermitian Clifford and $\Ot_\rho(A,B)=\langle V_ABV_AB\rangle_\rho$ is a Pauli expectation value. Writing $V_A$ as a symplectic matrix $M_A\in\mathrm{Sp}(2n,\F_2)$, $M_A^2=\1$ implies $(M_A-\1)^2=0$, so $M_A$ fixes at least $2^n$ Pauli strings up to sign, each contributing $|\Ot_\rho|=1$. This gives (SM, Sec. \ref{app:third})
\begin{equation}
    \I_\alpha(\rho,U)\le n\log2-\log\big(1+2^{-n}-4^{-n}\big)
    \label{eq:third}
\end{equation}
for all $\rho$, all $\alpha$ and all $U$ in the third level: such gates carry at most half of the maximal Pauli instability~\eqref{eq:range}.

\emph{Random-matrix dynamics.}---As an example of chaotic Hamiltonian evolution we take $U(t)=e^{-iWt}$ with $W$ an $N\times N$ Wigner matrix, $N=2^n$, normalized to a semicircle on $[-2,2]$. Cipolloni, Erd\H{o}s and Henheik~\cite{Cipolloni2024} computed the OTOC of deterministic observables to leading order in $N$: for traceless Pauli strings $A\neq B$ at infinite temperature,
\begin{equation}
    \Ot_{\1/N}(A,B)=\chi(A,B)\,\varphi(t)^4+O(\varepsilon),\quad \varphi(t)=\frac{J_1(2t)}{t},
\end{equation}
with $\varepsilon\simeq(t^4)/N$ and $J_1$ the Bessel function. The Pauli average is then immediate (SM, Sec.~\ref{app:wigner}):
\begin{equation}
    \I_2\big(\1/N,U(t)\big)=-8\log|\varphi(t)|+O\big(\varphi^{-4}\varepsilon,\,\varphi^{-8}N^{-2}\big).
    \label{eq:wigner}
\end{equation}
Magic grows as $\I_2\simeq4t^2$ at early times and, since $|\varphi(t)|\sim t^{-3/2}$, logarithmically, $\I_2\simeq12\log t$, at intermediate times, until Eq.~\eqref{eq:range} is approached. Remarkably, the leading-order result is the same for every pair $(A,B)$, so the Pauli \emph{average} is not needed for random-matrix dynamics. In the Gibbs state $\rho_\beta\propto e^{-\beta W}$ the real part of the OTOC is known~\cite{Cipolloni2024}, $\mathrm{Re}\,\Ot_{\rho_\beta}(A,B)=\chi(A,B)\varphi(t)^3\,\mathrm{Re}\,\varphi(t+i\beta)/\varphi(i\beta)+\dots$, which yields the upper bound $\I_2(\rho_\beta,U(t))\le-\log\E_{A,B}(\mathrm{Re}\,\Ot)^2$, in particular $\I_2\le C(\beta)t^2$ at early times with $C(\beta)=7-6I_0(2\beta)/[\beta I_1(2\beta)]+6/\beta^2$ interpolating between $C(0)=4$ and $C(\infty)=7$.

\emph{Heisenberg XXX model.}---In our final example, we analyze the Pauli instability of the time evolution operator $U(t)=e^{-iHt}$ for integrable and non-integrable forms of the Heisenberg XXX spin-chain model. Defined in terms of nearest-neighbor (NN) interactions, the integrable $n$ qubit XXX model is given by
\begin{equation}
    H_{\text{int}} = J_1 \sum_{i=1}^n \vec{\sigma}_i \cdot \vec{\sigma}_{i+1}
\end{equation}
where $\vec{\sigma}_i = (X_i,Y_i,Z_i)$ denotes the Pauli operators which act on the $i$-th qubit, $J_1$ is the NN coupling strength, and there are periodic boundary conditions (PBC) such that $\vec{\sigma}_{n+i} = \vec{\sigma}_i$. The integrability of this system can be broken by introducing a deformation term $H_{\text{non-int}} = H_{\text{int}} + H_{\text{def}}$. One particular choice, often referred to as the $J_1 - J_2$ model, is to introduce next-nearest-neighbor (NNN) interactions
\begin{equation}
    H_{\text{def}} = J_2 \sum_{i=1}^n \vec{\sigma}_i \cdot \vec{\sigma}_{i+2}
\end{equation}
where $J_2$ is the NNN coupling strength and there is once again PBC. We chose to focus on the antiferromagnetic case ($J_1,J_2 > 0$) and work in the vicinity of the integrable case, setting $J_1 = 1$ and $J_2 = 1.1$. As depicted in Figure~\ref{fig:single_probe}, we plot $\I_2 (\rho,U;A)$ for a fixed probe $A$ as a function of time. The graphs $(i)$, $(ii)$, and $(iii)$ present the results for the maximally mixed state $\1 /d$, Gibbs state $e^{-\beta H} /Z$, and ground state $\left| \psi_0 \rangle \langle \psi_0 \right|$ respectively. Since $k \geq \I_2 (\rho,U;A)/\log 2$, we can take the maximum $\I_2 (\rho,U;A)$ found in each case and place a lower bound on the number of $T$ gates needed to simulate the dynamics of $U$ in the given state $\rho$. For each state, the result for the integrable and non-integrable models is the same: we find $k \geq 9$, $k \geq 5$, and $k \geq 4$ respectively. Interestingly, we observe larger fluctuations in the non-integrable case compared to the integrable one. 
For the standard OTOC with fixed operators 
the opposite situation has been observed.
\cite{Fortes_2019,otoc_residual}.
\begin{figure}[htbp]
    \centering
    \includegraphics[width=\columnwidth]{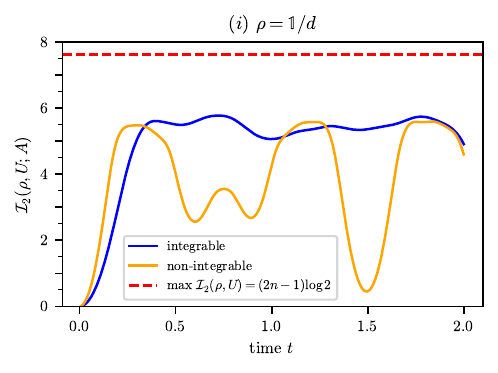}
    \includegraphics[width=\columnwidth]{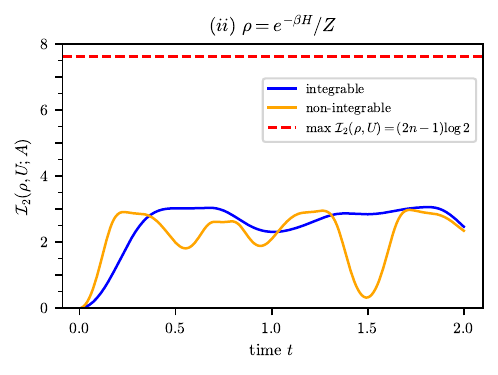}
    \includegraphics[width=\columnwidth]{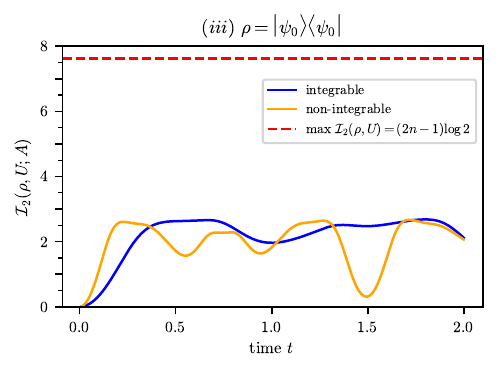}
    \caption{\emph{Estimate of magic via fixed probe operator.} These graphs were produced for $n=6$ qubits and the randomly chosen fixed probe $A = XIZYZZ$. Note also that $J_1 = 1$ and $J_2 = 1.1$. For the Gibbs state, we set $\beta = 1$. }
    \label{fig:single_probe}
\end{figure}

\emph{Discussion.}---We have shown that the Pauli instability, unlike the two-point stabilizer entropy, remains a faithful and computable measure of the magic of dynamics when the infinite-temperature trace is replaced by an expectation value in an arbitrary full-rank state, and that it bounds the $T$-count of the dynamics in any state. Several directions suggest themselves. The bound~\eqref{eq:tcount2} depends on $\rho$ only through the thermal two-point functions in~\eqref{eq:alpha2}, so different states can be used as different witnesses of the same $T$-count, and one may optimize over $\rho$. It is possible to find a counter-example to the naive guess that it is always maximized for $\rho \propto \1$. Conjecture~\ref{conj}, if established, would make the finite-temperature witness as sharp as the infinite-temperature one.
OTOC correlators were originally introduced as probes into scrambling \cite{Shenker:2013pqa,Hosur:2015ylk} and, under certain conditions, quantum chaos \cite{Larkin:1969abc,Kitaev:2014t1,Xu:2019lhc,Trunin:2023rwm,otoc_residual}, so it would be interesting to connect our results to the existing links \cite{Leone:2021oqe,Ahmadi:2022bkg} between magic and chaos, as well as the emergent notions of entanglement "in time" \cite{Brukner:2004egd,Fullwood:2022rjd,Cotler:2017anu,Doi:2023zaf,Milekhin:2025ycm}.
Given that the quantity we proposed can be computed in a finite-temperature state suggests several natural physical questions.
Finite temperature correlation functions often admit hydrodynamic representations at late times and it would interesting to understand if magic plays any role in that description or if there is a hydrodynamic description for the propagation of magic. Holographic systems typically admit gravitational duals only at low temperatures and it would be interesting to understand how to compute Pauli instability on the gravity side. Recently it was argued that magic is necessary for describing gravitational backreaction is these systems \cite{Cao:2023mzo,Cao:2024nrx,Cao:2026uoq}. Sachdev--Ye--Kitaev (SYK) model \cite{SachdevYe,kitaevfirsttalk,ms} can be a natural starting point.

\begin{acknowledgments}
\emph{Acknowledgments.}
We are grateful to Charles~Cao, Xuchen~Cao, Al~Shapere for comments.
\end{acknowledgments}

\bibliographystyle{apsrev4-2}
\bibliography{pauli_instability_paper}

\clearpage
\onecolumngrid
\begin{center}
{\large\bfseries Supplemental Material for ``Pauli instability in arbitrary states: detecting magic in physical correlators''}
\end{center}
\setcounter{equation}{0}
\setcounter{section}{0}
\renewcommand{\theequation}{S\arabic{equation}}
\renewcommand{\thesection}{S\arabic{section}}
\setcounter{secnumdepth}{2}

Throughout, $n$ is the number of qubits, $d=2^n$, $\Pn$ is the Pauli group for $n$ qubits, $A(t)\equiv U^\dagger AU$, $\Ot_\rho(A,B)=\Tr[\rho A(t)BA(t)B]$, and $\chi(Q,B)=\pm1$ according to whether $Q$ and $B$ commute or anticommute. We use repeatedly that $\frac1d\Tr(PQ)=\delta_{PQ}$ for $P,Q\in\Pn$, and that for a fixed non-identity $R\in\Pn$ exactly half of the strings $B\in\Pn$ anticommute with $R$, so that $\E_B\,\chi(R,B)=\delta_{R,\1}$ and hence $\E_B\,\chi(Q,B)\chi(Q',B)=\E_B\chi(QQ',B)=\delta_{QQ'}$.

\section{Elementary bounds and dependence on $\alpha$}\label{app:basic}

\begin{proof}[Proof of Proposition~\ref{prop:basic}]

 $0\le|\Ot_\rho(A,B)|\le1$, hence $\I_\alpha(\rho,U)\ge0$: $\langle X,Y\rangle_\rho\equiv\Tr[\rho X^\dagger Y]$ is a positive semidefinite Hermitian form on operators, so the Cauchy--Schwarz inequality $|\langle X,Y\rangle_\rho|^2\le\langle X,X\rangle_\rho\langle Y,Y\rangle_\rho$ holds. With $X=A(t)$ and $Y=BA(t)B$, both unitary, $\langle X,X\rangle_\rho=\langle Y,Y\rangle_\rho=\Tr\rho=1$ and $\langle X,Y\rangle_\rho=\Ot_\rho(A,B)$.

Upper bound in (i): for $A=\1$, $\Ot_\rho(\1,B)=\Tr[\rho B^2]=1$ for all $4^n$ strings $B$; for $B=\1$ and $A\neq\1$, $\Ot_\rho(A,\1)=\Tr[\rho A(t)^2]=1$ for the remaining $4^n-1$ strings $A$. All other terms are non-negative, so $\E_{A,B}|\Ot_\rho|^\alpha\ge(2\cdot4^n-1)/4^{2n}$.

Property (ii): since $|\Ot_\rho|\le1$, $|\Ot_\rho|^{\alpha'}\le|\Ot_\rho|^{\alpha}$, giving the first inequality. For the second, $(\E|\Oc_\rho|^\alpha)^{1/\alpha}$ is non-decreasing in $\alpha$ (Jensen's inequality for the convex function $x^\alpha \mapsto (x^\alpha)^{\alpha'/\alpha}$), so $-\frac1\alpha\log\E|\Oc_\rho|^\alpha\ge-\frac1{\alpha'}\log\E|\Oc_\rho|^{\alpha'}$.

For (iii), let $V_1,V_2$ be Clifford and set $U'=V_1UV_2$. Then $U'^\dagger AU'=V_2^\dagger\,U^\dagger A'U\,V_2$ with $A'=V_1^\dagger AV_1\in\pm\Pn$, and
\beq
\Ot_\rho^{U'}(A,B)=\Tr\big[\rho\,V_2^\dagger U^\dagger A'U V_2\,B\,V_2^\dagger U^\dagger A'UV_2\,B\big]
=\Tr\big[(V_2\rho V_2^\dagger)\,U^\dagger A'U\,B'\,U^\dagger A'U\,B'\big]
=\Ot^{U}_{V_2\rho V_2^\dagger}(A',B'),
\eeq
with $B'=V_2BV_2^\dagger\in\pm\Pn$. The maps $A\mapsto A'$, $B\mapsto B'$ are bijections of $\Pn$ up to signs, which drop out of $|\Ot|$; hence $\I_\alpha(\rho,V_1UV_2)=\I_\alpha(V_2\rho V_2^\dagger,U)$.

For (iv), let $\rho=\rho_1\otimes\rho_2$, $U=U_1\otimes U_2$ and $A=A_1\otimes A_2$, $B=B_1\otimes B_2$, the OTOC factorizes, $\Ot_\rho(A,B)=\Ot_{\rho_1}(A_1,B_1)\Ot_{\rho_2}(A_2,B_2)$, and so does the uniform average of $|\Ot_\rho|^\alpha$ over $\Pn=\mathcal P_{n_1}\times\mathcal P_{n_2}$. Taking $-\log$ gives additivity.

The non-monotonicity example in the main text is obtained from Eq.~\eqref{eq:S-parallelT} with $k=1$: $\I_\alpha(\rho,T)=-\log\big(\tfrac34+\tfrac14|\langle Z\rangle_\rho|^\alpha\big)$, which is $0$ for $\rho=|0\rangle\langle0|$ and $\log\frac43$ for $\rho=|{+}\rangle\langle{+}|$.

\end{proof}

\section{Faithfulness}\label{app:faith}

\begin{lemma}\label{lem:innerproduct}
If $\rho$ is invertible, $\langle X,Y\rangle_\rho=\Tr[\rho X^\dagger Y]$ is a positive-definite Hermitian inner product on the space of operators, and $|\langle X,Y\rangle_\rho|^2=\langle X,X\rangle_\rho\langle Y,Y\rangle_\rho$ holds iff $X$ and $Y$ are linearly dependent.
\end{lemma}
\begin{proof}
Hermiticity is immediate. $\langle X,X\rangle_\rho=\Tr[(X\rho^{1/2})^\dagger(X\rho^{1/2})]\ge0$ with equality iff $X\rho^{1/2}=0$, i.e.\ iff $X=0$ when $\rho^{1/2}$ is invertible. The equality case of Cauchy--Schwarz for a positive-definite form is standard.
\end{proof}

\begin{lemma}\label{lem:separate}
For distinct $P,P'\in\Pn$ there exists $B\in\Pn$ with $\chi(P,B)\neq\chi(P',B)$.
\end{lemma}
\begin{proof}
$PP'$ is proportional to a non-identity Pauli string, so some $B$ anticommutes with it, and $\chi(P,B)\chi(P',B)=\chi(PP',B)=-1$.
\end{proof}

\begin{proof}[Proof of Theorem~1]
If $U$ is Clifford, $A(t)=\pm P$ for some $P\in\Pn$, so $A(t)BA(t)B=PBPB=\chi(P,B)\1$ and $|\Ot_\rho(A,B)|=1$ for all $A,B$ and all $\rho$; thus $\I_\alpha=0$.

Conversely, suppose $\I_\alpha(\rho,U)=0$ with $\rho$ invertible. Since every $|\Ot_\rho|\le1$, the average can only equal one if $|\Ot_\rho(A,B)|=1$ for all $A,B$. By Lemma~\ref{lem:innerproduct} applied to the unit vectors $X=A(t)$, $Y=BA(t)B$, this means $BA(t)B=c_{AB}A(t)$ with $|c_{AB}|=1$. Conjugating once more by $B$ gives $A(t)=c_{AB}BA(t)B=c_{AB}^2A(t)$, so $c_{AB}=\pm1$: every $A(t)$ either commutes or anticommutes with every Pauli string $B$. Expand $A(t)=\sum_Pc_PP$. Then $BA(t)B=\sum_Pc_P\chi(P,B)P=\pm A(t)$ requires $\chi(P,B)$ to take the same value for all $P$ in the support of $A(t)$, for every $B$. By Lemma~\ref{lem:separate} the support therefore consists of a single string $P_A$, and $A(t)=c\,P_A$ with $c=\pm1$ by unitarity and Hermiticity. Since this holds for every $A\in\Pn$, $U$ maps the Pauli group to itself under conjugation, i.e.\ $U$ is Clifford up to a global phase.

The argument is local in $A$, which gives the second statement of Theorem~3(i): for full-rank $\rho$ and a fixed probe $A$, $\E_B|\Ot_\rho(A,B)|^2=1$ (i.e.\ $\I_2(\rho,U;A)=0$) iff $|\Ot_\rho(A,B)|=1$ for all $B$ iff $U^\dagger AU\in\pm\Pn$, by the same steps.

Without full rank the argument only shows $BA(t)B=\pm A(t)$ on the support of $\rho$. The example $\rho=|0\rangle\langle0|$, $U=T$, for which Eq.~\eqref{eq:S-parallelT} below gives $\I_\alpha=0$, shows that faithfulness can indeed fail.
\end{proof}

\section{Exact $B$-average at $\alpha=2$ and the Choi state}\label{app:Baver}

Let $A(t)=\sum_Qc_QQ$, $c_Q=c_Q(A)=\frac1d\Tr[QA(t)]\in\mathbb R$ (real because $A(t)$ and $Q$ are Hermitian), and define $m_Q\equiv\Tr[\rho\,A(t)\,Q]$. Then $BA(t)B=\sum_Qc_Q\chi(Q,B)Q$ and
\beq
\Ot_\rho(A,B)=\sum_Qc_Q\,\chi(Q,B)\,m_Q,\qquad
|\Ot_\rho(A,B)|^2=\sum_{Q,Q'}c_Qc_{Q'}\chi(Q,B)\chi(Q',B)\,m_Q\overline{m_{Q'}} .
\eeq
Averaging over $B$ with $\E_B\chi(Q,B)\chi(Q',B)=\delta_{QQ'}$ gives Eq.~\eqref{eq:alpha2} of the main text,
\begin{equation}
\E_B|\Ot_\rho(A,B)|^2=\sum_Qc_Q(A)^2\,|m_Q|^2,\qquad m_Q=\langle A(t)Q\rangle_\rho .
\label{eq:S-alpha2}
\end{equation}
(One may equally write $|\langle QA(t)\rangle_\rho|=|m_Q|$, since $\langle QA(t)\rangle_\rho=\overline{\langle A(t)Q\rangle_\rho}$.) Two identities will be used below:
\begin{equation}
\sum_Qc_Q\,m_Q=\Tr[\rho A(t)^2]=1,\qquad
\sum_Q|m_Q|^2=d\,\Tr\rho^2 ,
\label{eq:S-sumrules}
\end{equation}
the second following from the completeness relation $\sum_Q\Tr[XQ]\Tr[QY]=d\,\Tr[XY]$ applied to $X=\rho A(t)$, $Y=A(t)\rho$; it shows that the thermal amplitudes $m_Q$ are systematically larger for purer states.

At infinite temperature $m_Q=c_Q$ and $\E_B|\Ot|^2=\sum_Qc_Q^4$. The Choi state $\Choi U=(U\otimes\1)\Choi\Phi$, $\Choi\Phi=d^{-1/2}\sum_i|i\rangle^* |i\rangle$, has Pauli expectation values $\langle\!\langle U|A\otimes Q^{T}|U\rangle\!\rangle=\frac1d\Tr[U^\dagger AUQ]=c_Q(A)$, and $Q\mapsto Q^T$ is a bijection of $\Pn$ up to signs. Hence
\beq
\E_A\sum_Qc_Q(A)^4=\frac1{4^n}\sum_{P\in\mathcal P_{2n}}\langle\!\langle U|P|U\rangle\!\rangle^4=e^{-M_2(\Choi U)},
\eeq
with $M_2(\psi)=-\log\frac1{D}\sum_{P}\langle\psi|P|\psi\rangle^4$ the stabilizer 2-R\'enyi entropy on $D=4^n$ dimensions~\cite{Leone2022}. This is Eq.~\eqref{eq:choi}.

\section{The $T$-count bound via operator size}\label{app:op_size}

\emph{Normal form.} Let $U=C_{k+1}T_{q_k}C_k\cdots T_{q_1}C_1$ with Cliffords $C_j$ and $T_{q}=e^{-i\pi Z_q/8}$ (up to a phase) on qubit $q$. Since $VT_qV^\dagger=e^{-i\pi VZ_qV^\dagger/8}$, moving all Cliffords to the left gives
\beq
U=C\,e^{-i\pi P_k/8}\cdots e^{-i\pi P_1/8},\qquad C=C_{k+1}\cdots C_1,\qquad P_j=(C_j\cdots C_1)^\dagger Z_{q_j}(C_j\cdots C_1)\in\pm\Pn .
\eeq
The strings $P_j$ need not be distinct or independent.

\emph{Support.} For a Pauli strings $Q,P$: $e^{i\theta P}Qe^{-i\theta P}=Q$ if $[Q,P]=0$ and $=\cos(2\theta)Q+i\sin(2\theta)PQ$ if $\{Q,P\}=0$; in both cases the result is supported on $\{Q,PQ\}$. Since $A(t)=e^{i\pi P_1/8}\cdots e^{i\pi P_k/8}(C^\dagger AC)e^{-i\pi P_k/8}\cdots e^{-i\pi P_1/8}$ and $C^\dagger AC\in\pm\Pn$, induction over the $k$ conjugations shows that $\supp A(t)$ is contained (up to signs) in the coset $(C^\dagger AC)\cdot\langle P_1,\dots,P_k\rangle$ of the abelian group generated by the $P_j$ in $\Pn\simeq\F_2^{2n}$. Its size is $2^r$, $r=\operatorname{rank}_{\F_2}\{P_j\}\le\min(k,2n)$. Hence $N_A\le2^r\le2^k$.

\emph{Cauchy--Schwarz.} By \eqref{eq:S-sumrules}, $1=\big|\sum_{Q\in\supp A(t)}c_Qm_Q\big|^2\le N_A\sum_Qc_Q^2|m_Q|^2=N_A\,\E_B|\Ot_\rho(A,B)|^2$, using \eqref{eq:S-alpha2}. Thus $\E_B|\Ot_\rho(A,B)|^2\ge1/N_A$ for every $A$ and every $\rho$, and averaging over $A$,
\beq
e^{-\I_2(\rho,U)}=\E_A\E_B|\Ot_\rho|^2\ge\E_A\frac1{N_A}\ge2^{-r}\ge2^{-k}.
\eeq
This proves Theorem \ref{thm:tcount}; the intermediate statement $\E_B|\Ot_\rho(A,B)|^2\ge1/N_A$, is valid for every fixed $A$. Taking into account that $N_A \le 2^r \le 2^k$, covers Theorem \ref{thm:fixedA}.
Nothing in the argument depends on the rotation angle, so the same bound holds for $k$ arbitrary Pauli rotations $e^{i\theta_jP_j}$.

\emph{Refinement for independent axes.} We bound $\E_AN_A$ more carefully. Identify Pauli strings (mod signs) with vectors $a\in\F_2^{2n}$ and let $\omega(a,b)\in\F_2$ be the symplectic form, so that $\omega(a,b)=1$ iff the strings anticommute; let $v_j$ be the vector of $P_j$ and $a$ that of $C^\dagger AC$. Define $\hat T_v(S)=S\cup\{b+v:\,b\in S,\ \omega(b,v)=1\}$ for $S\subseteq\F_2^{2n}$; by the support rule above, $\supp A(t)\subseteq S^{(k)}(a)$ where $S^{(0)}=\{a\}$ and $S^{(j)}=\hat T_{v_j}(S^{(j-1)})$ (the order in which the $v_j$ are applied is immaterial for what follows). Call $I=\{i_1<\dots<i_m\}\subseteq[k]$ \emph{admissible for $a$} if $\omega\big(a+\sum_{i\in I,\,i<i_l}v_i,\ v_{i_l}\big)=1$ for $l=1,\dots,m$. By induction on $j$, $S^{(j)}(a)=\{a+\sum_{i\in I}v_i:\ I\subseteq[j]\ \text{admissible for }a\}$: the new elements produced at step $j$ are $b+v_j$ with $b=a+\sum_{i\in I}v_i$, $I\subseteq[j-1]$ admissible, and $\omega(b,v_j)=1$, i.e.\ exactly those with $I\cup\{j\}$ admissible. Hence $N_A\le\#\{I\subseteq[k]\ \text{admissible for }a\}$. For fixed $I$ with $|I|=m$, admissibility is the affine system $\omega(a,v_{i_l})=1+\sum_{i\in I,i<i_l}\omega(v_i,v_{i_l})$, $l=1,\dots,m$. If the $v_j$ are linearly independent, the map $a\mapsto(\omega(a,v_{i_1}),\dots,\omega(a,v_{i_m}))$ is surjective ($\omega$ is non-degenerate), so the system has exactly $2^{2n-m}$ solutions $a$. Therefore
\beq
\E_a\#\{I\ \text{admissible}\}=\sum_{m=0}^k\binom km2^{-m}=\Big(\frac32\Big)^k ,
\eeq
and by Jensen's inequality $\E_A N_A^{-1}\ge\big(\E_AN_A\big)^{-1}\ge(2/3)^k$, i.e.\ $\I_2(\rho,U)\le k\log\frac32$. (For dependent axes the affine systems can have $0$ or $2^{2n-\mathrm{rank}}>2^{2n-m}$ solutions, and this counting does not close; the coset bound $2^r$ then applies.)

\section{Infinite temperature tight bound}
\label{app:proof_of_infT}
Without loss of generality, let us  assume that $U = U' T$, where $T$ acts on the first qubit. We are going to show the following:
\begin{equation}
\label{eq:peeling}
    e^{-\I_2(\1/d,  U' T)} \ge  \frac{3}{4}e^{-\I_2(\1/d, U')}
\end{equation}
Then Theorem \ref{thm:tcount_tight} will follow from the repeated application of this inequality plus the fact that the value of $\I_2$ does not change if one removes a Clifford gate at the end of $U$.
\begin{proof}[Proof of (\ref{eq:peeling})]
    Lets denote $A' = U'^\dagger A U'$. Then using eq. (\ref{eq:choi})
    \begin{equation}
e^{-\I_2(\1/d,  U' T)} = \E_{A,B}|\Ot_{\1/d}(A,B)|^2=\frac{1}{4^n} \sum_{Q,A} c_Q(T^\dagger A' T)^4
\label{eq:1decomp}
\end{equation}
In turn, $A'$ can be expanded in Pauli string basis. We can separate contributions having a specific Pauli acting on the first qubit:
\begin{equation}
    A' = \sum_{Q_r} (c'_{\1 Q_r} \1 + c'_{Z Q_r} Z + c'_{X Q_r} X + c'_{Y Q_r} Y) \otimes Q_r, 
\end{equation}
where $Q_r$ is a Pauli string which acts on the rest of the qubits.
Conjugating with $T$ can be done explicitly:
\begin{equation}
    T^{\dagger} A' T = \sum_{Q_r} \left[ c'_{\1 Q_r} \1 + c'_{Z Q_r} Z + \frac{1}{\sqrt{2}}(c'_{X Q_r} + c'_{Y Q_r}) X + \frac{1}{\sqrt{2}}(c'_{Y Q_r} - c'_{X Q_r}) Y  \right] \otimes Q_r
\end{equation}
Plugging this inside eq. (\ref{eq:1decomp}) produces the following expression:
\begin{align}
    e^{-\I_2(\1/d,  U' T)} = \sum_{A,Q_r} \left[ c'^4_{\1 Q_r} + c'^4_{Z Q_r} + \frac{1}{2} c'^4_{X Q_r}  + \frac{1}{2} c'^4_{Y Q_r} + 3 c'^2_{X Q_r} c'^2_{Y Q_r}  \right] \ge \frac{3}{4}\sum_{A,Q_r} \left[ c'^4_{\1 Q_r} + c'^4_{Z Q_r} + c'^4_{X Q_r}  + c'^4_{Y Q_r}  \right] + \nonumber \\ 
    + \frac{1}{4}\sum_{A,Q_r} \left[ c'^4_{\1 Q_r} + c'^4_{Z Q_r} - c'^4_{X Q_r}  - c'^4_{Y Q_r}  \right].
    \end{align}
    The first sum is simply $e^{-\I_2(\1/d,U')}$.
The crux of the present argument is that the last sum is positive. 
Upto an overall constant, it can be rewritten as the following overlap in 4 copies of the system:
\begin{equation}
\label{eq:pauli_sums}
    \sum_{A,Q} s_1(Q) \Tr(Q^{\otimes 4} (U^\dagger)^{\otimes 4} A^{\otimes 4} U^{\otimes 4}), 
 \end{equation}
 where auxillary function $s_1(Q)$ is equal to $1$ is the first Pauli of $Q$ is $\1$ or $Z$ and $-1$ otherwise. Sum $\sum_A A^{\otimes 4}$ is a projector upto a positive constant. It can also be easily verified that $\sum_Q s_1(Q) Q^{\otimes 4}$ is a projector too, upto a positive constant. Hence, eq. (\ref{eq:pauli_sums}) is simply overlap of two projectors which is a positive number. This concludes the proof.
\end{proof}
Note that in the present argument we explicitly used the summation over $A$, this is why for fixed $A$ the tight bound is different as discussed in the main text. Also, it is tempting to try to prove the "peeling" inequality (\ref{eq:peeling}) for $\rho \neq \1/d$. However, it is possible to find a counter-example in this case. So general $\rho$ would require a different proof strategy.  

\section{Parallel $T$ gates}\label{app:Tpar}

Let $U=T^{\otimes k}\otimes\1^{\otimes(n-k)}$. Single-qubit algebra gives $T^\dagger\1T=\1$, $T^\dagger ZT=Z$, $T^\dagger XT=(X-Y)/\sqrt2$, $T^\dagger YT=(X+Y)/\sqrt2$, and one checks that for single-qubit Paulis $a,b$ the operator $(T^\dagger aT)\,b\,(T^\dagger aT)\,b$ equals $\pm\1$ for $12$ of the $16$ pairs $(a,b)$ and $\pm iZ$ for the $4$ pairs $a\in\{X,Y\}$, $b\in\{X,Y\}$. For $A=\bigotimes_ia_i$, $B=\bigotimes_ib_i$ the operator $A(t)BA(t)B$ is therefore a phase times $Z_S=\prod_{i\in S}Z_i$, where $S\subseteq[k]$ is the set of $T$-qubits on which $(a_i,b_i)$ is one of the four ``$Z$'' pairs; the qubits without $T$ gates contribute $\pm\1$. Under the uniform measure on $(A,B)$ each qubit belongs to $S$ independently with probability $1/4$, so
\begin{equation}
\E_{A,B}|\Ot_\rho(A,B)|^\alpha=\sum_{S\subseteq[k]}\Big(\frac14\Big)^{|S|}\Big(\frac34\Big)^{k-|S|}\big|\langle Z_S\rangle_\rho\big|^\alpha
=\Big(\frac34\Big)^k\sum_{S\subseteq[k]}3^{-|S|}\big|\langle Z_S\rangle_\rho\big|^\alpha ,
\label{eq:S-parallelT}
\end{equation}
which is Eq.~\eqref{eq:parallelT}. The last sum is $\ge1$ (the $S=\emptyset$ term) with equality iff $\langle Z_S\rangle_\rho=0$ for all non-empty $S$, and $\le\sum_S3^{-|S|}=(4/3)^k$ with equality iff $|\langle Z_S\rangle_\rho|=1$ for all $S$, i.e.\ iff the reduced state on the $k$ qubits is a product of $|0\rangle$'s and $|1\rangle$'s.

\emph{Proof of Theorem~3(iii).} For a fixed probe the same bookkeeping applies without the average over $A$. Take $A=X^{\otimes k}$ on the $T$-qubits (and $\1$ elsewhere). For each $T$-qubit, $b_i\in\{\1,Z\}$ contributes $\pm\1$ and $b_i\in\{X,Y\}$ contributes $\pm iZ$, so under the uniform measure on $B$ each qubit belongs to $S$ independently with probability $1/2$ and
\beq
e^{-\I_2(\rho,T^{\otimes k};X^{\otimes k})}=\E_B|\Ot_\rho(X^{\otimes k},B)|^2=2^{-k}\sum_{S\subseteq[k]}\big|\langle Z_S\rangle_\rho\big|^2 .
\eeq
Here $A(t)=2^{-k/2}\bigotimes_i(X_i-Y_i)$ has $N_A=2^k$, so at infinite temperature, or whenever $\langle Z_S\rangle_\rho=0$ for all non-empty $S$, the Cauchy--Schwarz inequality $\E_B|\Ot_\rho|^2\ge1/N_A$ is saturated and $\I_2(\rho,U;A)=k\log2$: the per-probe bound \eqref{eq:fixedAbound} is tight. This is Eq.~\eqref{eq:fixedAtight}. Averaging over $A$ instead gives $(3/4)^k$, because a fraction $1/2$ of the probes commutes with each $T$-axis and is not split.

\section{Third level of the Clifford hierarchy}\label{app:third}

Let $U^\dagger\Pn U\subseteq\mathcal C_n$. For each $A$, $V_A\equiv A(t)$ is a Clifford unitary with $V_A^2=\1$, so $V_ABV_A=\pm Q_{A,B}$ with $Q_{A,B}\in\Pn$, and $\Ot_\rho(A,B)=\pm\langle Q_{A,B}B\rangle_\rho$. Whenever $Q_{A,B}=B$ this equals $\pm1$. In the symplectic representation $V_A\leftrightarrow M_A\in\mathrm{Sp}(2n,\F_2)$, the strings fixed up to sign are $\ker(M_A-\1)$, and $M_A^2=\1$ gives $(M_A-\1)^2=M_A^2+\1=0$ over $\F_2$, so $\operatorname{im}(M_A-\1)\subseteq\ker(M_A-\1)$ and rank--nullity yields $\dim\ker(M_A-\1)\ge n$. Thus for every $A\neq\1$ at least $2^n-1$ non-identity strings $B$ give $|\Ot_\rho(A,B)|=1$. Counting as in the proof of Proposition~\ref{prop:basic}(i), and discarding the remaining non-negative terms,
\beq
\E_{A,B}|\Ot_\rho|^\alpha\ge\frac{4^n+(4^n-1)+(4^n-1)(2^n-1)}{4^{2n}}=2^{-n}\big(1+2^{-n}-4^{-n}\big),
\eeq
for every $\rho$ and $\alpha$, which is Eq.~\eqref{eq:third}.

\section{Wigner-matrix dynamics}\label{app:wigner}

Let $W$ be an $N\times N$ Wigner matrix (i.i.d. centered entries, variance $1/N$) and $U(t)=e^{-iWt}$. For deterministic $A,B$ with $\langle X\rangle\equiv\frac1N\Tr X$, Ref.~\cite{Cipolloni2024} shows, with $\varphi(t)=J_1(2t)/t$ and with very high probability,
\begin{equation}
\tfrac12\big\langle|[A(t),B]|^2\big\rangle
=\langle A^2\rangle\langle B^2\rangle\big[1-\varphi^2\big]+2\langle AB\rangle^2\varphi^2\big[\varphi(2t)-\varphi^2\big]+\langle A^2B^2\rangle\varphi^2-\langle ABAB\rangle\varphi^4+O(\varepsilon),
\end{equation}
where $\varphi\equiv\varphi(t)$ and, for Pauli strings, $\varepsilon\simeq(t^4+t)N^{-1}e^{t/N^{1/2-\delta}}$. Note that the error term has a precise probabilistic definition known as the stochastic domination and is discussed in Ref.~\cite{Cipolloni2024}, but this is beyond what we are trying to show through this example. Expanding $|[A,B]|^2=2-BABA-ABAB$ one finds $\frac12\langle|[A(t),B]|^2\rangle=1-\Ot_{\1/N}(A,B)$, and for Pauli strings $\langle A^2\rangle=\langle B^2\rangle=\langle A^2B^2\rangle=1$, $\langle AB\rangle=\delta_{AB}$, $\langle ABAB\rangle=\chi(A,B)$. Hence
\beq
\Ot_{\1/N}(A,B)=\chi(A,B)\varphi^4-2\delta_{AB}\,\varphi^2\big[\varphi(2t)-\varphi^2\big]+O(\varepsilon)\qquad(A,B\neq\1).
\eeq
The pairs with $A=B$ or with $A=\1$ or $B=\1$ are a fraction $O(4^{-n})$ of all pairs, so $\E_{A,B}|\Ot_{\1/N}|^2=\varphi(t)^8+O(\varphi^4\varepsilon+\varepsilon^2)+O(N^{-2})$, which gives Eq.~\eqref{eq:wigner}. Using $\varphi(t)=1-t^2/2+O(t^4)$ and $\varphi(t)\simeq\pi^{-1/2}t^{-3/2}\cos(2t-3\pi/4)$ for $t\gg1$, one obtains $\I_2\simeq4t^2$ at early times and an envelope $\I_2\simeq12\log t+\log\pi^4$ at intermediate times, valid as long as $\varphi(t)^4\gg\max(\varepsilon,N^{-1})$; the zeros of $J_1$ are regularized by the subleading terms.

At finite temperature, with $\rho_\beta=e^{-\beta W}/\Tr e^{-\beta W}$, Ref.~\cite{Cipolloni2024} gives $\frac12\Tr[\rho_\beta|[A(t),B]|^2]=1-\mathrm{Re}\,\Ot_{\rho_\beta}(A,B)$ (the same manipulation as above, now $\Tr[\rho BABA]=\overline{\Tr[\rho ABAB]}$) in terms of the analytic continuation $\varphi(z)$, $\varphi(i\beta)=I_1(2\beta)/\beta$; for Pauli strings $A\neq B$,
\beq
\mathrm{Re}\,\Ot_{\rho_\beta}(A,B)=\chi(A,B)\,\varphi(t)^3\,\frac{\mathrm{Re}\,\varphi(t+i\beta)}{\varphi(i\beta)}+O(\varepsilon).
\eeq
Since $|\Ot|^2\ge(\mathrm{Re}\,\Ot)^2$, this yields the upper bound $\I_2(\rho_\beta,U(t))\le-6\log|\varphi(t)|-2\log|\mathrm{Re}\,\varphi(t+i\beta)/\varphi(i\beta)|+\dots$. Writing $g(\beta)=\varphi(i\beta)$ and using that $\varphi$ is even and real on the real axis, $\mathrm{Re}\,\varphi(t+i\beta)=g(\beta)-\frac{t^2}2g''(\beta)+O(t^4)$, so at early times $\I_2(\rho_\beta,U(t))\le C(\beta)t^2+O(t^4)$ with
\beq
C(\beta)=3+\frac{g''(\beta)}{g(\beta)}=7-\frac{6\,I_0(2\beta)}{\beta\,I_1(2\beta)}+\frac{6}{\beta^2},
\eeq
which decreases from $C(\infty)=7$ to $C(0)=4$, matching the infinite-temperature coefficient.

\end{document}